\documentclass[letterpaper,11pt]{article}
\usepackage{amsmath,mathrsfs,amssymb,amsfonts,amsthm,mathtools}
\usepackage{physics,color}
\usepackage{graphicx}
\usepackage{subfigure}
\usepackage{tabularx} 
\usepackage[table]{xcolor}
\usepackage{tikz,tikz-cd}
\usetikzlibrary{shapes.geometric, arrows, positioning}
\usepackage{empheq}
\usepackage{bm}

\newtheorem{lemma}{Lemma}
\newtheorem{remark}{Remark}
\newcommand*\widefbox[1]{\fbox{\hspace{2em}#1\hspace{2em}}}
\pdfoutput=1 
\usepackage{jheppub} 
\usepackage[T1]{fontenc} 

\newcommand{\nn}{\nonumber}

\hypersetup{
    colorlinks=true,
    linkcolor=blue,
    citecolor=blue,
    urlcolor=blue
}

\title{Unitary reformulation of the thermofield double state and limits of cyclic multi-mode squeezing}

\author[a]{Arash Azizi}

\affiliation[a]{{\it The Institute for Quantum Science and Engineering,
Texas A\&M University,\\ College Station, TX 77843, U.S.A.}}

\emailAdd{sazizi@tamu.edu}

\abstract{
We investigate the structure and uniqueness of squeezed vacuum states defined by annihilation conditions of the form $(a - \alpha a^\dagger)|\psi\rangle = 0$ and their multimode generalizations, with applications to the Thermofield Double (TFD) state in quantum field theory. For $N=1$ and $N=2$, we demonstrate that these conditions uniquely define the single- and two-mode squeezed vacua, generated by unitary squeezing operators. A key result is the unitary reformulation of the TFD state, expressed as a product of two-mode squeezing operators, ensuring invertibility and resolving the non-unitary paradox in the Minkowski--Rindler vacuum correspondence. Extending to cyclic annihilation conditions $(a_i - \alpha_i a_{i+1}^\dagger)|\psi\rangle = 0$ with $a_{N+1} \equiv a_1$, we find that non-trivial squeezed states exist only for $N=2$. For $N > 2$, we establish a no-go theorem, proving no normalizable, non-trivial solutions exist, revealing a fundamental limit on cyclic multi-mode entanglement. These results highlight the bipartite nature of TFD-like entanglement and constrain multipartite generalizations in multi-region quantum field theories.
}

\begin{document} 
\maketitle
\flushbottom

\section{Introduction}

The Thermofield Double (TFD) state is a cornerstone in modern theoretical physics, bridging quantum field theory in curved spacetime, black hole thermodynamics, and quantum information theory through its elegant representation of thermal states as pure entangled states in a doubled Hilbert space \cite{Unruh1976, Israel76}. A prominent example is the Minkowski vacuum state expressed in terms of modes localized in two distinct, causally disconnected regions—often represented as the left and right Rindler wedges. In this two-mode representation, the Minkowski vacuum state $\ket{0_{M}}$ is annihilated by the following coupled mode conditions:
\begin{equation}
\bigl(b_{R\,\omega}- e^{-\frac{\pi \omega}{a}}\,b_{L\,\omega}^\dagger\bigr)\ket{0_M}=0\,, \quad
\bigl(b_{L\,\omega}- e^{-\frac{\pi \omega}{a}}\,b_{R\,\omega}^\dagger\bigr)\ket{0_M}=0\,,
\label{brblMinkvac}
\end{equation}
where $b_{R\,\omega}^{\dagger}$ and $b_{L\,\omega}^{\dagger}$ are creation operators associated with field modes localized respectively in the right and left Rindler wedges. Physically, the right and left wedges correspond to distinct, causally disconnected regions accessible to uniformly accelerated observers. Such accelerated observers perceive the Minkowski vacuum $\ket{0_M}$ as a thermally populated state, characterized by the Unruh temperature $T_U=\frac{a}{2\pi}$, with $a$ denoting the observers' proper acceleration. Thus, the frequency-dependent exponential factor $e^{-\frac{\pi \omega}{a}}$ in Eq.~(\ref{brblMinkvac}) encodes this thermal distribution, illustrating the inevitable entanglement between modes localized in separate wedges.

The formal solution of these annihilation conditions explicitly demonstrates the structure of quantum entanglement induced by relativistic acceleration:
\begin{equation}
\ket{0_{M}}=\frac{1}{\sqrt{Z}}\,
\exp\left[\int_0^{\infty} d\omega\, e^{-\frac{\pi \omega}{a}}\,b_{L \omega}^{\dagger}b_{R \omega}^{\dagger}\right] \ket{0_{L}}\otimes\ket{0_{R}}\,,
\label{MinkRind}
\end{equation}
where $\ket{0_L}$ and $\ket{0_R}$ represent vacuum states defined by observers restricted entirely within their respective Rindler wedges. This representation clearly highlights that the Minkowski vacuum is a highly nontrivial entangled state when expressed in the accelerated (Rindler) frame, analogous to the structure of the Thermofield Double (TFD) states frequently encountered in quantum gravity and holography. 

While this expression captures thermal and entanglement structures elegantly, it is generated by an operator that is inherently non-unitary. Although any two states in a Hilbert space are trivially related by some unitary operator, the conceptual challenge here is more specific. The task is not to argue for the mere existence of such a transformation, but rather to provide an explicit \textbf{construction} that reconciles the widely-used, physically-motivated non-unitary generator in Eq.~(\ref{MinkRind}) with a physically realizable and invertible transformation at the level of individual modes. Clarifying this is fundamental for conceptual completeness and for ensuring the reversibility of the quantum transformations involved.

To explore this question systematically, we draw on the powerful mathematical and conceptual tools of quantum optics. In particular, the squeezing operator introduced by Caves \cite{caves1981quantum}, which generates single- and two-mode squeezed states of light, plays an essential role in our analysis. Squeezed states are characterized by reduced quantum noise in one quadrature at the expense of its conjugate and have become vital resources for precision measurements, such as gravitational wave detection in LIGO \cite{LIGO2016, aasi2013enhanced, Oelker:14}, and quantum information applications including quantum key distribution, computing, and teleportation \cite{braunstein2005quantum, Menicucci2006, ohliger2010limitations, vanloock1999, weedbrook2012gaussian}. Defined as vacua for appropriately transformed annihilation operators, squeezed states yield unique single-mode squeezed vacua ($N=1$) and two-mode squeezed vacua ($N=2$), the latter being primary sources of continuous-variable entanglement \cite{walls1983squeezed, loudon1987squeezed, mandel1995optical}.

A central result of this work is the successful unitary reformulation of the TFD state. We demonstrate that the Minkowski vacuum can be generated from the Rindler vacuum via a product of unitary two-mode squeezing operators, as given in Eq.~\eqref{eq:TFD_Unitary}. This ensures invertibility and resolves the apparent non-unitary paradox by establishing a clear algebraic framework for the Minkowski–Rindler correspondence. 

This insight motivates a natural question: can this pairwise entanglement structure be extended to a multipartite scenario involving $N$ distinct regions? Such a generalization is motivated by analogous looped or cyclic entanglement structures that are fundamental in both condensed matter and high-energy theory. In condensed matter physics, for instance, ring-shaped quantum spin chains and periodic tensor network states are widely studied as models for strongly correlated systems~\cite{Affleck1987, Verstraete2008, Cirac2009, Evenbly2009}. In the context of holography and AdS/CMT~\cite{Sachdev2011}, multi-boundary wormhole geometries are conjectured to be dual to multipartite entangled states, a concept central to the ER=EPR correspondence and its extensions~\cite{Maldacena2013ER=EPR, Hartman2013, Balasubramanian2014, Bhattacharya2020}. Motivated by these physical and geometric analogies, we test the most direct mathematical generalization of the bipartite annihilation conditions by proposing a cyclically coupled $N$-mode condition of the form:
\begin{equation}
(a_i - \alpha_i a_{i+1}^\dagger)\ket{\psi}=0\,,\quad\text{with cyclic identification}\quad a_{N+1}\equiv a_1\,.
\end{equation}
Remarkably, we uncover a fundamental structural limitation. While nontrivial squeezed states arise naturally for $N=1$ and $N=2$, we prove that no nontrivial, normalizable solution exists for any number of modes $N>2$. The straightforward cyclic generalization of the two-mode entangled vacuum fails, establishing a sharp no-go theorem for this class of multipartite squeezed vacua. This result underscores the privileged status of bipartite entanglement structures like the Minkowski-Rindler vacuum and cautions against naive extensions to more general multipartite frameworks.

We summarize our key findings in Table~\ref{tab:squeezing_summary}. The paper is organized as follows: Section~\ref{sec:one_mode} rigorously demonstrates the uniqueness and unitary generation of single-mode squeezed states; Section~\ref{sec:two_modes} extends this analysis to two-mode squeezed vacua; finally, Section~\ref{sec:n_modes} provides a detailed derivation and discussion of the no-go theorem for cyclically coupled multipartite squeezed states.

\begin{table}
\centering
\caption{Summary of uniqueness and existence of squeezed states $\ket{\psi}$ for different numbers of modes under specified annihilation constraints.}
\label{tab:squeezing_summary}
\renewcommand{\arraystretch}{1.6}
\begin{tabularx}{\textwidth}{|c|X|X|}
\hline
\textbf{Modes} & \textbf{Annihilation Condition(s)} & \textbf{Form of $\ket{\psi}$ and Existence} \\
\hline

\textbf{1}
& \begin{center}
    $(a - \alpha a^\dagger)\ket{\psi} = 0$
\end{center}
& Unique solution:
\[
\ket{\psi} = (1 - |\alpha|^2)^{1/4} \sum_{n=0}^{\infty} \alpha^n \frac{(2n-1)!!}{(2n)!!} \ket{2n}
\]
Valid for $|\alpha| < 1$; corresponds to standard single-$N$-mode squeezed vacuum, generated by a unitary operator. \\

\hline

\textbf{2}
& \begin{center}
    $\begin{cases}
(a - \alpha b^\dagger)\ket{\psi} = 0 \\
(b - \beta a^\dagger)\ket{\psi} = 0
\end{cases}$
\end{center}
& Unique solution \newline($\alpha$ should be equal to $\beta$):
\[
\ket{\psi} = \sqrt{1 - |\alpha|^2} \sum_{m=0}^{\infty} \alpha^m \ket{m,m}
\]
Valid for $|\alpha| < 1$; matches standard two-$N$-mode squeezed vacuum, generated by a unitary operator. \\

\hline

\textbf{$N > 2$}
& \begin{center}
    $(a_i - \alpha_i a_{i+1}^\dagger)\ket{\psi} = 0$ \newline for $i=1,\dots,N$ (cyclic: $a_{N+1} \equiv a_1$)
  \end{center}
&  \begin{center}
    \textbf{No solution exists unless all $\alpha_i = 0$}. (i.e., reduces to vacuum state)
   \end{center} \\
\hline
\end{tabularx}
\end{table}

\section{Single-Mode Squeezing and Uniqueness}\label{sec:one_mode}

\subsection{Uniqueness of the Squeezed Vacuum State}

We investigate whether the unitary squeezing operator
\begin{align}
S(\xi) = e^{\frac{1}{2} ( \xi^* a^2 - \xi a^{\dagger 2})}
\end{align}
produces the same state as the non-unitary exponential $e^{\frac{1}{2} \alpha a^{\dagger 2}}$ when acting on the vacuum, i.e., whether:
\begin{align}
S(\xi) \ket{0} \stackrel{?}{=} e^{\frac{1}{2} \alpha a^{\dagger 2}} \ket{0}, \label{eq:relation}
\end{align}
where $\xi = r e^{i\theta}$. This addresses whether a state defined by annihilation conditions, often leading to a non-unitary form in the Fock basis, can be generated by a unitary transformation, ensuring physical implementability and interconvertibility.

We establish the uniqueness of the state $\ket{\psi_\alpha}$ defined by:
\begin{equation}
(a - \alpha a^\dagger) \ket{\psi_\alpha} = 0.
\label{eq:squeezing_condition}
\end{equation}
Representing $\ket{\psi_\alpha} = \sum_{n=0}^{\infty} C_n \ket{n}$ in the Fock basis and substituting into Eq.~\eqref{eq:squeezing_condition}, we obtain:
\begin{align}
(a - \alpha a^\dagger) \sum_{n=0}^{\infty} C_n \ket{n} = \sum_{n=0}^{\infty} C_n \sqrt{n} \ket{n-1} - \alpha \sum_{n=0}^{\infty} C_n \sqrt{n+1} \ket{n+1} = 0.
\end{align}
Equating coefficients of $\ket{n}$, we derive the recurrence relation for $n \ge 1$:
\begin{equation}
C_{n+1} \sqrt{n+1} = \alpha C_{n-1} \sqrt{n}.
\label{eq:recurrence}
\end{equation}
For $n=0$, the $\ket{0}$ coefficient gives $C_1 = 0$, implying $C_{2n+1} = 0$ for all $n \ge 0$. Iterating Eq.~\eqref{eq:recurrence} for even indices, we find:
\begin{equation}
C_{2n} = \alpha^n \sqrt{\frac{(2n-1)!!}{(2n)!!}} C_0.
\label{eq:C2n_formula}
\end{equation}
To normalize $\ket{\psi_\alpha}$, we compute:
\begin{equation}
\bra{\psi_\alpha} \ket{\psi_\alpha} = |C_0|^2 \sum_{n=0}^{\infty} |\alpha|^{2n} \frac{(2n-1)!!}{(2n)!!} = |C_0|^2 (1 - |\alpha|^2)^{-1/2} = 1,
\end{equation}
using the generating function $(1-x)^{-1/2} = \sum_{n=0}^{\infty} \frac{(2n-1)!!}{(2n)!!} x^n$ (see Appendix~\ref{app:lemma}). Thus, $|C_0|^2 = (1 - |\alpha|^2)^{1/2}$, and choosing $C_0 = (1 - |\alpha|^2)^{1/4}$ (real and positive), the unique state is:
\begin{empheq}[box=\widefbox]{equation}
(a - \alpha a^\dagger) \ket{\psi_\alpha} = 0
\quad \text{implies} \quad
\ket{\psi_\alpha} = (1 - |\alpha|^2)^{1/4} \sum_{n=0}^{\infty} \alpha^n \frac{(2n-1)!!}{(2n)!!} \ket{2n}.
\label{eq:psi_alpha_unique}
\end{empheq}
\vspace{.5 cm}
\subsubsection*{Condition for $|\alpha| < 1$}
From Eq.~\eqref{eq:squeezing_condition}, $a \ket{\psi_\alpha} = \alpha a^\dagger \ket{\psi_\alpha}$. Taking norms and using $[a, a^\dagger] = 1$, we get:
\begin{equation}
\langle \psi_\alpha | a^\dagger a | \psi_\alpha \rangle = |\alpha|^2 (\langle \psi_\alpha | a^\dagger a | \psi_\alpha \rangle + 1).
\end{equation}
Letting $\bar{n} = \langle \psi_\alpha | a^\dagger a | \psi_\alpha \rangle$, this simplifies to $\bar{n} = |\alpha|^2 (\bar{n} + 1)$, yielding:
\begin{equation}
\bar{n} = \frac{|\alpha|^2}{1 - |\alpha|^2}.
\end{equation}
Since $\bar{n} \ge 0$, we require $1 - |\alpha|^2 > 0$, implying $|\alpha| < 1$ for a normalizable state.

\subsection{Squeezing Operator}

The uniqueness of the state defined by $(a - \alpha a^\dagger) \ket{\psi_\alpha} = 0$ implies that the standard squeezed vacuum state, generated by the unitary squeezing operator $S(\xi) = e^{\frac{1}{2} ( \xi^* a^2 - \xi a^{\dagger 2})}$ acting on $\ket{0}$, takes the form derived in Eq.~\eqref{eq:psi_alpha_unique}. The squeezed vacuum $S(\xi)\ket{0}$ satisfies (see, e.g., \cite{azizi.sq.coh}):
\begin{equation}
(a + \tanh r \, e^{i\theta} a^\dagger) S(\xi) \ket{0} = 0,
\end{equation}
where $\xi = r e^{i\theta}$. Comparing with Eq.~\eqref{eq:squeezing_condition}, we identify:
\begin{equation}
\boxed{\quad \alpha = -\tanh r \, e^{i\theta}. \quad}
\end{equation}
The normalization factor is:
\begin{equation}
1 - |\alpha|^2 = \frac{1}{\cosh^2 r}.
\end{equation}
Thus, the squeezed vacuum state is:
\begin{empheq}[box=\widefbox]{equation}
S(\xi) \ket{0} = \frac{1}{\sqrt{\cosh r}} \sum_{n=0}^{\infty} (\tanh r)^n \, e^{in(\theta+\pi)} \sqrt{\frac{(2n-1)!!}{(2n)!!}} \ket{2n}.
\end{empheq}

To confirm, the non-unitary state $e^{\frac{1}{2} \alpha a^{\dagger 2}} \ket{0}$ satisfies Eq.~\eqref{eq:squeezing_condition}, as $[a, e^{\frac{1}{2} \alpha a^{\dagger 2}}] = \alpha a^\dagger e^{\frac{1}{2} \alpha a^{\dagger 2}}$ (see Appendix~\ref{app:commutator}) implies $(a - \alpha a^\dagger) e^{\frac{1}{2} \alpha a^{\dagger 2}} \ket{0} = 0$. Given the uniqueness of the solution (Eq.~\eqref{eq:psi_alpha_unique}), the states are equivalent, confirming:
\begin{align}
\boxed{\quad S(\xi) \ket{0} = (1 - |\alpha|^2)^{1/4}
e^{\frac{1}{2} \alpha a^{\dagger 2}} \ket{0}, \quad \text{with} \quad \alpha = -\tanh r \, e^{i\theta}. \quad}
\label{eq:relation2}
\end{align}
This demonstrates that the non-unitary exponential form is equivalent to a state generated by a unitary operator for the single-mode case.

\section{Two-Mode Squeezing: Coupled Annihilation Conditions}\label{sec:two_modes}

\subsection{Uniqueness of the Two-Mode Squeezed Vacuum State}

Consider two modes satisfying $a \ket{0}_a = 0$, $b \ket{0}_b = 0$, with shorthand notation $\ket{m,n} \equiv \ket{m_a} \otimes \ket{n_b}$. We seek the unique state $\ket{\psi}$ defined by:
\begin{align}
(a - \alpha b^\dagger) \ket{\psi} &= 0, \label{eq:coupled_a} \\
(b - \beta a^\dagger) \ket{\psi} &= 0, \label{eq:coupled_b}
\end{align}
and determine if it can be generated by a unitary operator, extending the single-mode analysis.

Expanding $\ket{\psi} = \sum_{m,n=0}^{\infty} C_{mn} \ket{m,n}$ in the Fock basis and substituting into Eq.~\eqref{eq:coupled_a}, we obtain the recurrence relation:
\begin{equation}
C_{m+1,n} = \alpha \sqrt{\frac{n}{m+1}} C_{m,n-1}.
\label{eq:recurrence_a_two_mode}
\end{equation}
Similarly, Eq.~\eqref{eq:coupled_b} yields:
\begin{equation}
C_{m,n+1} = \beta \sqrt{\frac{m}{n+1}} C_{m-1,n}.
\label{eq:recurrence_b_two_mode}
\end{equation}
Applying these relations to coefficients $C_{m,0}$ ($m \ge 1$) and $C_{0,n}$ ($n \ge 1$), we find $C_{m,0} = 0$ and $C_{0,n} = 0$. Off-diagonal coefficients $C_{m,n}$ ($m \neq n$) vanish (see Appendix~\ref{app:off_diagonal}), implying only $C_{k,k}$ are non-zero. Consistency requires $\alpha = \beta$, simplifying the recurrence to:
\begin{equation}
C_{m+1,m+1} = \alpha C_{m,m} \quad \Rightarrow \quad C_{m,m} = \alpha^m C_{0,0}.
\end{equation}
Thus, $\ket{\psi} = C_{0,0} \sum_{m=0}^{\infty} \alpha^m \ket{m,m}$. Normalizing, $\bra{\psi} \ket{\psi} = |C_{0,0}|^2 \frac{1}{1 - |\alpha|^2} = 1$, gives $C_{0,0} = \sqrt{1 - |\alpha|^2}$. The unique state is:
\begin{empheq}[box=\widefbox]{align}
(a - \alpha b^\dagger) \ket{\psi} = 0 \quad \text{and} \quad
(b - \beta a^\dagger) \ket{\psi} = 0 \quad \text{implies} \nn\\
\alpha = \beta, \quad \text{and} \quad
\ket{\psi} = \sqrt{1 - |\alpha|^2} \sum_{m=0}^{\infty} \alpha^m \ket{m,m}. \label{eq:psi_two_mode_summary}
\end{empheq}

\subsection{Transformation of Operators Using the Two-Mode Squeezing Operator}

We consider the unitary two-mode squeezing operator:
\begin{equation}
T(\xi) = e^{\xi a^\dagger b^\dagger - \xi^* ab},
\label{eq:two_mode_squeezing_operator}
\end{equation}
where \(\xi = r e^{i\theta}\). Using the Baker-Campbell-Hausdorff formula, \( e^X Y e^{-X} = Y + [X,Y] + \frac{1}{2!}[X,[X,Y]] + \cdots \), with \( X = \xi a^\dagger b^\dagger - \xi^* ab \) and \( Y = a \), we compute \( T a T^\dagger \). Key commutators are:
\begin{align}
[X, a] &= -\xi b^\dagger, \quad [X, -\xi b^\dagger] = |\xi|^2 a = r^2 a.
\end{align}
Summing the series:
\begin{equation}
T a T^\dagger = a \cosh r - \xi b^\dagger \frac{\sinh r}{r} = \cosh r \, a - e^{i\theta} \sinh r \, b^\dagger.
\label{eq:TaTdagger}
\end{equation}
Similarly:
\begin{equation}
T b T^\dagger = \cosh r \, b - e^{i\theta} \sinh r \, a^\dagger.
\label{eq:TbTdagger}
\end{equation}
Since \( T a \ket{0,0} = (\cosh r \, a - e^{i\theta} \sinh r \, b^\dagger) T \ket{0,0} = 0 \) (as \( a \ket{0,0} = 0 \)), we obtain:
\begin{equation}
(a - e^{i\theta} \tanh r \, b^\dagger) T \ket{0,0} = 0.
\label{eq:tmsv_diff_eq_a}
\end{equation}
Similarly:
\begin{equation}
(b - e^{i\theta} \tanh r \, a^\dagger) T \ket{0,0} = 0.
\label{eq:tmsv_diff_eq_b}
\end{equation}
Comparing with Eqs.~\eqref{eq:coupled_a} and \eqref{eq:coupled_b}, we identify \( \alpha = \beta = \tanh r \, e^{i\theta} \). The unique state from Eq.~\eqref{eq:psi_two_mode_summary} is:
\begin{equation}
\ket{\psi} = \sqrt{1 - |\alpha|^2} \sum_{n=0}^{\infty} \alpha^n \ket{n,n} = \frac{1}{\cosh r} \sum_{n=0}^{\infty} (\tanh r)^n e^{in\theta} \ket{n,n},
\end{equation}
matching \( \ket{\psi} = \sqrt{1 - |\alpha|^2} e^{\alpha a^\dagger b^\dagger} \ket{0,0} \). Verifying, since \( [a, e^{\alpha a^\dagger b^\dagger}] = \alpha b^\dagger e^{\alpha a^\dagger b^\dagger} \), we have:
\begin{equation}
(a - \alpha b^\dagger) e^{\alpha a^\dagger b^\dagger} \ket{0} = e^{\alpha a^\dagger b^\dagger} a \ket{0} = 0.
\end{equation}
Similarly, \( (b - \alpha a^\dagger) e^{\alpha a^\dagger b^\dagger} \ket{0} = 0 \). By uniqueness, we conclude:
\begin{align}
\label{eq:relation_two_mode}
\boxed{\quad e^{\xi a^\dagger b^\dagger - \xi^* ab} \ket{0} = \sqrt{1 - |\alpha|^2} \, e^{\alpha a^\dagger b^\dagger} \ket{0}, \quad \text{with} \quad \xi = r e^{i\theta} \quad \text{and} \quad \alpha = \tanh r \, e^{i\theta}, \quad}
\end{align}
The state \( T(\xi)\ket{0,0} \) is normalized due to unitarity, confirming that the two-mode squeezed vacuum is generated by a unitary operator, ensuring physical realizability.
\begin{remark}
In this section, the two-mode squeezed state is defined as \( T(\xi) \ket{0} = e^{\xi a^\dagger b^\dagger - \xi^* ab} \ket{0} \), whereas in Section~\ref{sec:one_mode}, the single-mode squeezed state uses the convention \( S(\xi) \ket{0} = e^{\frac{1}{2} (\xi^* a^2 - \xi a^{\dagger 2})} \ket{0} \) (cf.~\cite{scully1997quantum}). The differing exponent structures result in a minus sign discrepancy in the parameter \(\alpha\): \(\alpha = \tanh r e^{i\theta}\) here versus \(\alpha = -\tanh r e^{i\theta}\) in Section~\ref{sec:one_mode}.
\end{remark}
\subsection{Relation to the Thermofield Double State and Invertibility}

The single- and two-mode squeezed vacua derived above clarify a key issue in relativistic quantum field theory concerning the Thermofield Double state, often expressed as:
\begin{align}
\ket{0_{M}} = \frac{1}{\sqrt{Z}} \exp\left( \int_0^\infty d\omega\, e^{-\beta \omega / 2}\, b_{L\omega}^\dagger b_{R\omega}^\dagger \right) \ket{0_L} \otimes \ket{0_R}, \label{eq:TFD}
\end{align}
where \( b_{L\omega}^\dagger \), \( b_{R\omega}^\dagger \) are creation operators for left and right Rindler wedges, and \( \beta = 2\pi/a \) is the inverse Unruh temperature. This non-unitary form suggests a non-invertible transformation between Rindler and Minkowski vacua, complicating the entanglement interpretation.

Based on Eq.~\ref{eq:relation_two_mode}, we have:
\begin{align}
\alpha = e^{-\beta \omega / 2} = \tanh r \quad \Rightarrow \quad \xi = r = \frac{1}{2} \ln \coth \left( \frac{\beta \omega}{4} \right),
\end{align}
hence, the TFD state can be written formally as a unitary transformation:
\begin{align}
\boxed{\quad \ket{0_{M}} = \exp\left\{ \frac{1}{2} \int_0^\infty d\omega\, \ln \coth \left( \frac{\beta \omega}{4} \right) \left( b_{L\omega}^\dagger b_{R\omega}^\dagger - b_{L\omega} b_{R\omega} \right) \right\} \ket{0_L} \otimes \ket{0_R}, \quad} \label{eq:TFD_Unitary}
\end{align}
The unitary nature of the two-mode squeezing operator \( T(\xi) = e^{\xi a^\dagger b^\dagger - \xi^* ab} \) ensures that the transformation from the two-mode vacuum to the squeezed state is invertible, a critical property for physical realizability. Unlike the non-unitary operator \( e^{\alpha a^\dagger b^\dagger} \), which lacks an inverse, \( T(\xi) \) satisfies \( T^\dagger T = I \), allowing the squeezed state \( T(\xi) \ket{0,0} \) to be reversibly mapped back to \( \ket{0,0} \). This invertibility, as shown in Eq.~\eqref{eq:relation_two_mode}, resolves the apparent non-unitary paradox in the Thermofield Double state’s formal expression (Eq.~\eqref{eq:TFD}) by demonstrating that each mode’s transformation is unitary, providing a clear algebraic framework for the Minkowski–Rindler correspondence and enabling practical implementation in quantum systems.

Our results in Sections \ref{sec:one_mode} and \ref{sec:two_modes} resolve this for each frequency mode. For single-mode squeezing, the state defined by \( (a - \alpha a^\dagger)\ket{\psi} = 0 \) equals \( S(\xi)\ket{0} \), where \( S(\xi) \) is unitary. Similarly, for two-mode squeezing, \( (a - \alpha b^\dagger)\ket{\psi} = 0 \), \( (b - \alpha a^\dagger)\ket{\psi} = 0 \) yields a state generated by the unitary \( T(\xi) \). Thus, the TFD state’s non-unitary form is a product of unitary transformations per mode, ensuring invertibility and clarifying the Minkowski–Rindler correspondence.

This insight motivates exploring multimode squeezing for entanglement across multiple spacetime regions, investigated in Section~\ref{sec:n_modes}, where we find a no-go result for cyclic \( N \)-mode squeezed vacua when \( N > 2 \).

\section{General \texorpdfstring{\boldmath $N$}{N}-Mode Squeezing: Recurrence Structure and No-Go Result under Cyclic Coupling}\label{sec:n_modes}  

\subsection{Uniqueness of \texorpdfstring{\boldmath $N$}{N}-Mode Squeezed State}

We generalize single- and two-mode squeezed states to the \( N \)-mode squeezed state (NMSS), defined by cyclic coupling conditions:
\begin{equation}
(a_i - \alpha_i a_{i+1}^\dagger) |\psi \rangle = 0, \quad i = 1, 2, \dots, N,
\end{equation}
where \( a_{N+1} \equiv a_1 \). This structure extends pairwise interactions, inspired by multi-region quantum field theories, such as generalizations of the Minkowski-Rindler framework.

Express the state as a Fock state superposition:
\begin{equation}
|\psi \rangle = \sum_{m_1, \dots, m_N} C_{m_1, \dots, m_N} | m_1, \dots, m_N \rangle,
\end{equation}
where \( | m_1, \dots, m_N \rangle = | m_1 \rangle \otimes \dots \otimes | m_N \rangle \). Applying \( (a_i - \alpha_i a_{i+1}^\dagger) |\psi \rangle = 0 \), we obtain:
\begin{equation}
\sum_{m_1, \dots, m_N} C_{m_1, \dots, m_N} \left( \sqrt{m_i} | m_1, \dots, m_i - 1, \dots, m_N \rangle - \alpha_i \sqrt{m_{i+1} + 1} | m_1, \dots, m_{i+1} + 1, \dots, m_N \rangle \right) = 0.
\end{equation}

\subsection{Recurrence Relations for Coefficients}

Equating coefficients of each Fock state \( | m_1, \dots, m_i, \dots, m_{i+1}, \dots, m_N \rangle \), we derive the recurrence relation:
\begin{equation}
C_{m_1, \dots, m_i+1, \dots, m_N} = \alpha_i \sqrt{\frac{m_{i+1}}{m_i+1}} C_{m_1, \dots, m_i, \dots, m_{i+1}-1, \dots, m_N},
\end{equation}
or equivalently:
\begin{equation}
C_{m_1, \dots, m_N} = \alpha_i \sqrt{\frac{m_{i+1}}{m_i}} C_{m_1, \dots, m_i-1, m_{i+1}-1, \dots, m_N}.
\end{equation}
For \( m_{i+1} \ge m_i \), iterating \( m_i \) times reduces \( m_i \) to zero:
\begin{equation}
C_{m_1, \dots, m_N} = \alpha_i^{m_i} \sqrt{\binom{m_{i+1}}{m_i}} C_{m_1, \dots, 0, m_{i+1}-m_i, \dots, m_N}.
\end{equation}
For \( m_i \ge m_{i+1} \), iterating \( m_{i+1} \) times reduces \( m_{i+1} \) to zero:
\begin{equation}
C_{m_1, \dots, m_N} = \alpha_i^{-m_{i+1}} \sqrt{\binom{m_i+m_{i+1}-1}{m_{i+1}}} C_{m_1, \dots, m_i+m_{i+1}, 0, \dots, m_N}.
\end{equation}
Iterating across \( i = 1, 2, \dots, N-1 \), coefficients are expressed in terms of those with earlier modes set to zero, subject to cyclic constraints. These relations underpin the no-go result for \( N > 2 \), derived below.
\subsection{Structural Properties of the \texorpdfstring{\boldmath $N$}{N}-Mode Squeezed State}

The master equation \( a_i |\psi \rangle = \alpha_i a_{i+1}^\dagger |\psi \rangle \) implies that \( |\psi \rangle \) is a superposition of Fock states with even total particle number, characteristic of a squeezed vacuum. For non-zero coefficients \( C_{m_1, \dots, m_N} \), the recurrence relations require \( m_i \le m_{i+1} \) for \( i = 1, \dots, N-1 \). To see this, consider a term with \( m_i > 0 \), \( m_{i+1} = 0 \):
\begin{align}
a_i | \dots, m_i, 0_{i+1}, \dots \rangle &= \sqrt{m_i} | \dots, m_i-1, 0_{i+1}, \dots \rangle, \\
\alpha_i a_{i+1}^\dagger | \dots, m_i, 0_{i+1}, \dots \rangle &= \alpha_i | \dots, m_i, 1_{i+1}, \dots \rangle.
\end{align}
Since these states are orthogonal for \( m_i > 0 \), the master equation holds only if their coefficients vanish, implying \( m_i = 0 \) when \( m_{i+1} = 0 \), or generally \( m_i \le m_{i+1} \).

The cyclic condition \( a_{N+1} \equiv a_1 \) adds \( m_N \le m_1 \), yielding:
\begin{equation}
m_1 \le m_2 \le \dots \le m_N \le m_1.
\end{equation}
This requires equal excitations:
\begin{equation}
m_1 = m_2 = \dots = m_N = k.
\end{equation}
Thus, \( |\psi \rangle = \sum_k C_k |k, k, \dots, k \rangle \). For \( N > 2 \), substituting into \( (a_i - \alpha_i a_{i+1}^\dagger) |\psi \rangle = 0 \):
\begin{equation}
\sum_k C_k \left( \sqrt{k} | \dots, k-1_i, \dots, k, \dots \rangle - \alpha_i \sqrt{k+1} | \dots, k_i, \dots, k+1_{i+1}, \dots \rangle \right) = 0,
\end{equation}
the orthogonal states require \( C_k \sqrt{k} = 0 \), \( C_k \alpha_i \sqrt{k+1} = 0 \), forcing \( C_k = 0 \) for \( k > 0 \). Thus, \( |\psi \rangle \) reduces to the vacuum \( |0, \dots, 0 \rangle \). For \( N = 2 \), the states align, allowing non-zero \( C_{k,k} \), as in the two-mode squeezed vacuum.

This no-go result for \( N > 2 \) highlights a fundamental limit in extending pairwise unitary transformations, like Minkowski-Rindler correspondences, to cyclic \( N \)-mode systems using these annihilation conditions.

\subsection{Distinction from Decoupled Pairwise Entanglement}
It is crucial to distinguish our model from an alternative multipartite structure composed of a collection of $N$ separate two-mode TFD states, such as $|\text{TFD}_n\rangle \propto \sum_k e^{-\gamma_k} |k\rangle_n \otimes |k\rangle_{n+1}$. Each such state is a standard two-mode squeezed vacuum, satisfying its own pair of annihilation conditions (e.g., $(a_n - \alpha_n a_{n+1}^\dagger)|\text{TFD}_n\rangle = 0$ and $(a_{n+1} - \alpha_n a_{n}^\dagger)|\text{TFD}_n\rangle = 0$). This arrangement describes a chain of pairwise entanglements without a global cyclic constraint imposed on a single state. Our no-go theorem does not apply to this decoupled case. Instead, our result demonstrates that a \textit{single}, globally entangled $N$-mode state, defined by the stricter condition of simultaneous cyclic annihilation, cannot be formed for $N>2$. This highlights that a straightforward generalization of the TFD state's defining relations to a single, cyclically symmetric multipartite state fails, reinforcing the unique structural properties of bipartite entanglement in this context.
\section{Conclusion}
In this paper, we have systematically investigated the uniqueness and structural properties of squeezed states, defined by annihilation conditions in increasing complexity from one to $N$ modes. Our work addresses a critical theoretical challenge in quantum mechanics: establishing when states represented by non-unitary exponentials (like the Thermofield Double state) can be equivalently generated by unitary and invertible transformations.

For single-mode systems, we rigorously demonstrated that the state annihilated by $(a - \alpha a^\dagger)$ is uniquely determined up to a normalization constant. This state is shown to be equivalent to the standard single-mode squeezed vacuum state generated by a unitary and invertible squeezing operator, with a clear relationship established between the state parameter $\alpha$ and the operator parameter $\xi$. The non-negativity of the average particle number necessitates the condition $|\alpha| < 1$ for a normalizable state. This result is significant as it provides a rigorous demonstration of how states, which naturally appear in a non-unitary exponential form, can be understood as products of unitary transformations acting on the vacuum.

Extending our analysis to two-mode systems, we found that the state defined by coupled annihilation conditions $(a - \alpha b^\dagger) |\psi \rangle = 0$ and $(b - \beta a^\dagger) |\psi \rangle = 0$ is also uniquely determined. This uniqueness implies that the coefficients in the Fock basis are non-zero only for equal excitation numbers in both modes ($m=n$), requiring $\alpha = \beta$. This uniquely defined state is equivalent to the well-known two-mode squeezed vacuum state generated by the unitary two-mode squeezing operator, providing a consistent framework for understanding its properties and its generation via invertible transformations. These findings for $N=1$ and $N=2$ address a crucial theoretical aspect concerning the relationship between states defined by annihilation conditions and those generated by unitary operators, directly paralleling the transformation properties seen in the Thermofield Double state for two spacetime regions.

The investigation of $N$-mode squeezing under a cyclic coupling scheme, $(a_i - \alpha_i a_{i+1}^\dagger) |\psi \rangle = 0$, was motivated by exploring a direct generalization of the pairwise annihilation conditions inherent in structures like the TFD state for two spacetime regions. Our derivation of recurrence relations for the Fock coefficients showed that for non-trivial states to exist under this specific cyclic coupling, the mode excitation numbers must satisfy a chain of inequalities, $m_1 \le m_2 \le \dots \le m_N \le m_1$, which logically implies $m_1 = m_2 = \dots = m_N = k$. This suggested that the $N$-mode squeezed state would be a superposition of states where all modes possess an identical number of excitations, i.e., $|\psi \rangle = \sum_k C_k |k, k, \dots, k \rangle$.

However, a crucial implication arises for systems with $N>2$. By substituting this proposed form of the state back into the defining master equation, we found that the equation can only be satisfied if all coefficients $C_k$ for $k>0$ are zero. This leads to the unexpected conclusion that, for $N>2$, the $N$-mode squeezed state, as defined by the cyclic operator $(a_i - \alpha_i a_{i+1}^\dagger)$, reduces trivially (i.e., if all coefficients $\alpha_i=0$) to the vacuum state. The only exception is the $N=2$ case, where the standard two-mode squeezed vacuum state, characterized by identical excitation numbers in each mode, remains a valid non-trivial solution. This highlights that the specific form of the cyclic coupling operator, which affects only two adjacent indices at a time, imposes severe restrictions on the possibility of generating entangled $N$-mode squeezed states for $N>2$ in this framework. This sharp no-go result, particularly when viewed through analogies to multi-region quantum field theories such as the Minkowski-Rindler relation, implies that a direct, straightforward generalization of such pairwise unitary vacuum transformations to $N$-distinct spacetime regions via this specific cyclic coupling may not yield non-trivial squeezed states. It underscores a fundamental structural limitation beyond pairwise squeezing and provides important guidance for future theoretical explorations of multipartite entanglement generation in multi-mode systems.
\section*{Acknowledgments}

I am grateful to Girish Agarwal, Marlan Scully, and Bill Unruh for discussions. This work was supported by the
Robert A. Welch Foundation (Grant No. A-1261) and the National Science Foundation (Grant No. PHY-2013771).
\appendix

\section{Generating Function for Double Factorial Ratio}\label{app:lemma}
\begin{lemma}
The generating function for the double factorial ratio is:
\begin{equation}
(1-x)^{-1/2} = \sum_{n=0}^{\infty} \frac{(2n-1)!!}{(2n)!!} x^n.
\end{equation}
\end{lemma}

\begin{proof}
Consider the Taylor series of $f(x) = (1-x)^{-1/2}$ around $x=0$. The $n^{\text{th}}$ derivative is:
\begin{equation}
\frac{d^n}{dx^n} (1-x)^{-1/2} = \frac{(2n-1)!!}{2^n} (1-x)^{-(2n+1)/2}.
\end{equation}
At $x=0$, the Taylor coefficient is $a_n = \frac{1}{n!} \frac{(2n-1)!!}{2^n}$. Since $(2n)!! = 2^n n!$, we have $a_n = \frac{(2n-1)!!}{(2n)!!}$. Thus:
\begin{equation}
(1-x)^{-1/2} = \sum_{n=0}^{\infty} \frac{(2n-1)!!}{(2n)!!} x^n.
\end{equation}
(Note: $(-1)!! = 1$.)
\end{proof}

\section{Commutator Derivation}\label{app:commutator}

To verify that $e^{\frac{1}{2} \alpha a^{\dagger 2}} \ket{0}$ satisfies $(a - \alpha a^\dagger) \ket{\psi} = 0$, we compute the commutator $[a, e^{\frac{1}{2} \alpha a^{\dagger 2}}]$. Using the operator identity $[a, f(a^\dagger)] = f'(a^\dagger)$, for $f(a^\dagger) = e^{\frac{1}{2} \alpha a^{\dagger 2}}$, we have $f'(a^\dagger) = \alpha a^\dagger e^{\frac{1}{2} \alpha a^{\dagger 2}}$. Thus:
\begin{equation}
[a, e^{\frac{1}{2} \alpha a^{\dagger 2}}] = \alpha a^\dagger e^{\frac{1}{2} \alpha a^{\dagger 2}}.
\end{equation}
Applying this, we get:
\begin{align}
(a - \alpha a^\dagger) e^{\frac{1}{2} \alpha a^{\dagger 2}} \ket{0} &= \left( e^{\frac{1}{2} \alpha a^{\dagger 2}} a + [a, e^{\frac{1}{2} \alpha a^{\dagger 2}}] - \alpha a^\dagger e^{\frac{1}{2} \alpha a^{\dagger 2}} \right) \ket{0} \\
&= e^{\frac{1}{2} \alpha a^{\dagger 2}} a \ket{0} = 0,
\end{align}
since $a \ket{0} = 0$.

\section{Off-Diagonal Coefficients in Two-Mode Squeezing}\label{app:off_diagonal}

To show that off-diagonal coefficients $C_{m,n}$ ($m \neq n$) vanish for the state $\ket{\psi}$ satisfying $(a - \alpha b^\dagger) \ket{\psi} = 0$ and $(b - \beta a^\dagger) \ket{\psi} = 0$, consider the recurrence relations:
\begin{align}
C_{m+1,n} &= \alpha \sqrt{\frac{n}{m+1}} C_{m,n-1}, \\
C_{m,n+1} &= \beta \sqrt{\frac{m}{n+1}} C_{m-1,n}.
\end{align}
For $C_{m,0}$ ($m \ge 1$), Eq.~\eqref{eq:coupled_a} gives $C_{m+1,0} \sqrt{m+1} = 0$, so $C_{m,0} = 0$. Similarly, Eq.~\eqref{eq:coupled_b} gives $C_{0,n} = 0$ for $n \ge 1$. For any $C_{m,n}$ with $m \neq n$, repeatedly apply the recurrence relations. For example, if $m > n$, use the first relation to reduce $n$ until $n-1 < 0$, yielding $C_{m-n,k} = 0$. If $n > m$, use the second relation to reduce $m$ until $m-1 < 0$, yielding $C_{k,n-m} = 0$. Thus, all off-diagonal coefficients are zero, leaving only $C_{k,k}$ non-zero.

\bibliographystyle{jhep}
\bibliography{Squeezing}

\providecommand{\href}[2]{#2}\begingroup\raggedright\begin{thebibliography}{10}

\bibitem{Unruh1976}
W.~G. Unruh, {\it {Notes on black-hole evaporation}},  {\em Phys. Rev. D} {\bf 14} (1976) 870--892.

\bibitem{Israel76}
W.~Israel, {\it {Thermo-field dynamics of black holes}},  {\em Phys. Lett. A} {\bf 57} (1976) 107--110.

\bibitem{caves1981quantum}
C.~M. Caves, {\it {Quantum-mechanical noise in an interferometer}},  {\em Phys. Rev. D} {\bf 23} (Apr, 1981) 1693--1708.

\bibitem{LIGO2016}
B.~P. Abbott et~al., {\it Observation of gravitational waves from a binary black hole merger},  {\em Physical Review Letters} {\bf 116} (2016), no.~6 061102.

\bibitem{aasi2013enhanced}
J.~Aasi et~al., {\it Enhanced sensitivity of the ligo gravitational wave detector by using squeezed light},  {\em Nature Photonics} {\bf 7} (2013) 613--619.

\bibitem{Oelker:14}
E.~Oelker, L.~Barsotti, S.~Dwyer, D.~Sigg, and N.~Mavalvala, {\it Squeezed light for advanced gravitational wave detectors and beyond},  {\em Opt. Express} {\bf 22} (Aug, 2014) 21106--21121.

\bibitem{braunstein2005quantum}
S.~L. Braunstein and P.~van Loock, {\it {Quantum information with continuous variables}},  {\em Rev. Mod. Phys.} {\bf 77} (Jun, 2005) 513--577.

\bibitem{Menicucci2006}
N.~C. Menicucci, P.~van Loock, M.~Gu, C.~Weedbrook, T.~C. Ralph, and M.~A. Nielsen, {\it Universal quantum computation with continuous-variable cluster states},  {\em Phys. Rev. Lett.} {\bf 97} (Sep, 2006) 110501.

\bibitem{ohliger2010limitations}
M.~Ohliger, K.~Kieling, and J.~Eisert, {\it Limitations of quantum computing with gaussian cluster states},  {\em Phys. Rev. A} {\bf 82} (Oct, 2010) 042336.

\bibitem{vanloock1999}
P.~van Loock, {\it Optical hybrid approaches to quantum information},  {\em Laser \& Photonics Reviews} {\bf 5} (2011), no.~2 167--200, [\href{http://arxiv.org/abs/https://onlinelibrary.wiley.com/doi/pdf/10.1002/lpor.201000005}{{\tt https://onlinelibrary.wiley.com/doi/pdf/10.1002/lpor.201000005}}].

\bibitem{weedbrook2012gaussian}
C.~Weedbrook, S.~Pirandola, R.~Garc\'{\i}a-Patr\'on, N.~J. Cerf, T.~C. Ralph, J.~H. Shapiro, and S.~Lloyd, {\it {Gaussian quantum information}},  {\em Rev. Mod. Phys.} {\bf 84} (May, 2012) 621--669.

\bibitem{walls1983squeezed}
D.~F. Walls, {\it {Squeezed states of light}},  {\em nature} {\bf 306} (1983), no.~5939 141--146.

\bibitem{loudon1987squeezed}
R.~Loudon and P.~L. Knight, {\it Squeezed light},  {\em Journal of modern optics} {\bf 34} (1987), no.~6-7 709--759.

\bibitem{mandel1995optical}
L.~Mandel and E.~Wolf, {\em Optical Coherence and Quantum Optics}.
\newblock Cambridge University Press, Cambridge, 1995.

\bibitem{Affleck1987}
I.~Affleck and F.~D.~M. Haldane, {\it Critical theory of quantum spin chains},  {\em Phys. Rev. B} {\bf 36} (Oct, 1987) 5291--5300.

\bibitem{Verstraete2008}
F.~Verstraete, V.~Murg, and J.~Cirac, {\it Matrix product states, projected entangled pair states, and variational renormalization group methods for quantum spin systems},  {\em Advances in Physics} {\bf 57} (2008), no.~2 143--224, [\href{http://arxiv.org/abs/https://doi.org/10.1080/14789940801912366}{{\tt https://doi.org/10.1080/14789940801912366}}].

\bibitem{Cirac2009}
J.~I. Cirac and F.~Verstraete, {\it Renormalization and tensor product states in spin chains and lattices},  {\em Journal of Physics A: Mathematical and Theoretical} {\bf 42} (dec, 2009) 504004.

\bibitem{Evenbly2009}
G.~Evenbly and G.~Vidal, {\it Entanglement renormalization in two spatial dimensions},  {\em Phys. Rev. Lett.} {\bf 102} (May, 2009) 180406.

\bibitem{Sachdev2011}
S.~Sachdev, {\em Condensed Matter and AdS/CFT}, pp.~273--311.
\newblock Springer Berlin Heidelberg, Berlin, Heidelberg, 2011.

\bibitem{Maldacena2013ER=EPR}
J.~Maldacena and L.~Susskind, {\it Cool horizons for entangled black holes},  {\em Fortschritte der Physik} {\bf 61} (2013), no.~9 781--811, [\href{http://arxiv.org/abs/https://onlinelibrary.wiley.com/doi/pdf/10.1002/prop.201300020}{{\tt https://onlinelibrary.wiley.com/doi/pdf/10.1002/prop.201300020}}].

\bibitem{Hartman2013}
T.~Hartman and J.~Maldacena, {\it {Time evolution of entanglement entropy from black hole interiors}},  {\em Journal of High Energy Physics} {\bf 2013} (2013), no.~5 14.

\bibitem{Balasubramanian2014}
V.~Balasubramanian, P.~Hayden, A.~Maloney, D.~Marolf, and S.~F. Ross, {\it Multiboundary wormholes and holographic entanglement},  {\em Classical and Quantum Gravity} {\bf 31} (sep, 2014) 185015.

\bibitem{Bhattacharya2020}
A.~Bhattacharya, {\it {Multipartite purification, multiboundary wormholes, and islands in ${\mathrm{AdS}}_{3}/{\mathrm{CFT}}_{2}$}},  {\em Phys. Rev. D} {\bf 102} (Aug, 2020) 046013.

\bibitem{azizi.sq.coh}
A.~Azizi, {\it {Coherent and Squeezed States in Quantum Optics: From Traditional to Group-Theoretical Perspectives }},  {\em To appear} (2025).

\bibitem{scully1997quantum}
M.~O. Scully and M.~S. Zubairy, {\em Quantum Optics}.
\newblock Cambridge University Press, Cambridge, 1997.

\end{thebibliography}\endgroup
\end{document}